\documentclass[letterpaper,12pt]{article}
\pdfoutput=1

\usepackage{amssymb,amsmath,amsthm,amsfonts}
\usepackage{graphicx, xcolor, varwidth}
\usepackage[percent]{overpic}
\usepackage{mathtools}
\usepackage{setspace}
\usepackage{cite}
\usepackage{tikz,tkz-graph,tkz-berge}
\usepackage{simplewick}

\usetikzlibrary{decorations.markings}
\usetikzlibrary{shapes.geometric}
\pgfdeclarelayer{edgelayer}
\pgfdeclarelayer{nodelayer}
\pgfsetlayers{edgelayer,nodelayer,main}
\tikzstyle{none}=[inner sep=0pt] 
\tikzstyle{simplethree}=[circle,fill=black,draw=black,line width=1.500,minimum size = 2pt,inner sep = 2pt]
\tikzstyle{simple}=[circle,fill=black,draw=black,line width=1.000,minimum size = 2pt,inner sep = 2pt]
\tikzstyle{simpletoo}=[circle,fill=black,draw=black,line width=1.000,minimum size = 1pt,inner sep = 1pt]

\usepackage{xcolor}
\usepackage[dotinlabels]{titletoc}
\usepackage[colorlinks=true, urlcolor=dblue, linktocpage=true, linkcolor=dblue,citecolor=purple]{hyperref}

\newtheorem{defn}{Definition}

\newtheorem{comm}{Comment}
\newtheorem{thm}{Theorem}
\newtheorem{lemma}{Lemma}
\newtheorem{remark}{Remark}

\renewcommand{\ij}{{\langle i j \rangle}}

\newcommand{\Lap}{\bigtriangleup}

\newcommand{\mrm}{\mathrm}

\newcommand{\nn}{\nonumber}
\newcommand{\pd}{\partial}
\newcommand{\lb}{\left(}
\newcommand{\rb}{\right)}
\newcommand{\lcb}{\left\{}

\newcommand{\LL}{\mathcal{L}}
\newcommand{\MM}{\mathcal{M}}
\newcommand{\NN}{\mathcal{N}}
\newcommand{\OO}{\mathcal{O}}

\newcommand{\TT}{\mathcal{T}}

\newcommand{\rcb}{\right\}}
\newcommand{\lsb}{\left[}
\newcommand{\rsb}{\right]}

\newcommand\be{\begin{equation}}
\newcommand\ba{\begin{eqnarray}}
\newcommand\ee{\end{equation}}
\newcommand\ea{\end{eqnarray}}

\newcommand\padic{$p$-adic }

\numberwithin{equation}{section}

\colorlet{dblue}{blue!70!black}

\newcommand{\arxiv}[1]
  {\href{http://arxiv.org/abs/#1}{arXiv:#1}}

\newcommand\kl{{\langle kl \rangle}}

\usepackage{graphicx}
\usepackage{amsmath}

\begin{document}
\begin{spacing}{1.3}
\begin{titlepage}

\ \\
\vspace{-2.3cm}
\begin{center}

{\LARGE{General relativity from $p$-adic strings}}

\vspace{0.5cm}
An Huang,$^{1}$ Bogdan Stoica,$^{2}$ and Shing-Tung Yau$^{3,4}$

\vspace{5mm}

{\small
\textit{
$^1$Department of Mathematics, Brandeis University, Waltham, MA 02453, USA}\\

\vspace{2mm}

\textit{$^2$Department of Physics \& Astronomy, Northwestern University, Evanston, IL 60208, USA}\\

\vspace{2mm}

\textit{
$^3$Center of Mathematical Sciences And Applications, Harvard University, Cambridge MA 02138, USA }\\

\vspace{2mm}

\textit{
$^4$Department of Mathematics, Harvard University, Cambridge MA 02138, USA}\\

\vspace{4mm}

{\tt anhuang@brandeis.edu, bstoica@northwestern.edu, yau@math.harvard.edu}

\vspace{0.3cm}
}

\end{center}

\begin{abstract}
For an arbitrary prime number $p$, we propose an action for bosonic $p$-adic strings in curved target spacetime, and show that the vacuum Einstein equations of the target are a consequence of worldsheet scaling symmetry of the quantum $p$-adic strings, similar to the ordinary bosonic strings case. It turns out that spherical vectors of unramified principal series representations of $PGL(2,\mathbb{Q}_p)$ are the plane wave modes of the bosonic fields on $p$-adic strings, and that the regularized normalization of these modes on the $p$-adic worldsheet presents peculiar features which reduce part of the computations to familiar setups in quantum field theory, while also exhibiting some new features that make loop diagrams much simpler. Assuming a certain product relation, we also observe that the adelic spectrum of the bosonic string corresponds to the nontrivial zeros of the Riemann Zeta function.
\end{abstract}

\vfill

\begin{flushleft}{\small BRX-TH-6643, Brown HET-1778}\end{flushleft}

\end{titlepage}

\setcounter{tocdepth}{2}
\tableofcontents

\section{Introduction}
\label{intro}

The \padic string amplitudes, first proposed by Freund and Olson \cite{FreundOlson}, and then subsequently investigated in \cite{FreundWitten,BFOW} and many other works, are a staple subject in \padic physics. The invariant nature of these amplitudes across number fields when expressed in integral form, as well as the adelic relations between four-point functions \cite{FreundWitten}, and the more complicated relations recently explored in \cite{Bocardo-Gaspar:2017atv}, hint at deep connections between \padic and Archimedean physics. Relatedly, it was proven in \cite{2018arXiv180101189H} that on the $B$ side of mirror symmetry, important data about period integrals at the {\it large complex structure limit points} can be described in terms of \padic data and vice versa. Indeed, in the context of the program proposed in \cite{Stoica:2018zmi}, Archimedean string theory (and Archimedean results) should generally be obtainable from \padic strings. One of the motivations of the present paper thus is to make progress toward establishing a $p$-adic $B$~model.

In the Archimedean world, it is known that the target space Einstein equations follow from imposing that the  string two-point function, computed on the worldsheet, should not RG run at the one-loop (see e.g. \cite{tong}). This condition in turn follows from imposing that the metric of the target manifold should not RG flow.

In the present paper we are interested in studying a closely related question, that fits the general program outlined in \cite{Stoica:2018zmi} of obtaining Archimedean results from \padic physics: whether the usual Archimedean target space Einstein equations follow from the target manifold metric not running at one-loop, if the worldsheet is given by a $p$-adic string, instead of an Archimedean one. It turns out the answer is affirmative, and along the way we will encounter some nontrivial facts about the eigenfunctions of the Laplace operator on the Bruhat-Tits tree. Note that we are not asking whether the target space Einstein equations follow from considering all primes, but rather from considering any one fixed prime $p$. That one prime is sufficient is a feature of this particular question. When considering more involved computations, such as the string spectrum, analyzing more than one place will be required.

What do we mean by $p$-adic string? The old results \cite{FreundOlson,FreundWitten,BFOW} on amplitudes  do not require the existence of a worldsheet or of a sigma model, but, in the present paper, a sigma model and worldsheet are precisely what we are interested in. The sigma model action for $p$-adic strings has already been proposed by Zabrodin in \cite{zabrodin,zabrodin2}, 
\be
\label{zabact}
S \sim \sum_{\ij \in E(T_p)} \eta_{ab} \lb X^a_i - X^a_j\rb \lb X^b_i - X^b_j\rb .
\ee
Here $\eta_{ab}$ is the target space Minkowski metric, the worldsheet (at genus zero) is just the Bruhat-Tits tree $T_p$, and $\ij$ is the edge between vertices $i$ and $j$ (whenever $\ij$ is used $i$ and $j$ are implied to be neighbors).  Action \eqref{zabact} is thus analogous to the Polyakov action.

The Zabrodin action is defined when the target space is flat. In order to access the Ricci tensor of the target space, we need a suitable generalization to curved target space.  We propose that this generalization is given by
\be
\label{eq12}
S = \sum_{\ij \in E(T_p)} \frac{d^2(X_i,X_j)}{ V a^2_\ij},
\ee
so that the Zabrodin action \eqref{zabact} can be naturally understood as the long wavelength limit of Eq. \eqref{eq12}. Note that the sum in action \eqref{eq12} is naively divergent; much of this paper will be devoted to discussing its regularization.

We take Eq. \eqref{eq12} above as the definition of the \padic string action. In this expression $d(X_i,X_j)$ is the target space distance between points $X_i$, $X_j$, quantity $a_{\ij}$ is length of edge $\ij$ in the tree, and $V$ is the degree of any vertex (so that $V=p+1$ for the tree corresponding to $\mrm{SL}_2(\mathbb{Q}_p)$). The action is thus just a distance squared, and it is reminiscent of the Nambu-Goto action.

The plan for the rest of the paper is as follows. In Section \ref{planewaves}, we construct the momentum eigenfunctions on the Bruhat-Tits tree, discuss their normalization, and present a certain symmetry structure with which we equip the tree. We also briefly comment on the connection between the adelic string spectrum and the nontrivial zeros of the Riemann Zeta function. In Section \ref{stringsfromdistance}, we introduce the $p$-adic string action, and explain how it reduces to a Polyakov-like form in the weak target space curvature limit. Finally, in Section \ref{pathintsec} we compute the tree level and one loop correction to the two-point function. The momentum eigenfunction normalization is proven in Appendix \ref{appA}, and the vanishing of contributions from the real line to the tree level two-point function is shown in Appendix \ref{appvanishing}.

\section{Plane waves on the tree}
\label{planewaves}

The trees we are interested in are the Bruhat-Tits buildings for $\mrm{SL}_2\lb \mathbb{K}\rb$, with $\mathbb{K}$ a finite extension of $\mathbb{Q}_p$. In the rest of the paper we will only be concerned with the case $\mathbb{K}=\mathbb{Q}_p$; our results can be straightforwardly generalized to nontrivial finite extensions. For a more mathematical introduction to trees, see e.g. \cite{casselman,AbrBrown}. 

Combinatorially, the Bruhat-Tits tree $T_p$ for $\mrm{SL}_2\lb \mathbb{Q}_p\rb$ is an infinite tree of uniform degree  $p+1$, as in Figure \ref{bttree}. This space is congruent to $\mrm{PGL}\lb2, \mathbb{Q}_p \rb/\mrm{PGL}\lb2, \mathbb{Z}_p \rb$, just as in the Archimedean case the Poincar\'e disk is congruent to $\mrm{PGL}\lb 2,\mathbb{R}\rb/\mrm{PSO}(2)$ ($\mrm{PGL}\lb2, \mathbb{Z}_p \rb$ and $\mrm{PSO}(2)$ are the maximal compact subgroups in each case, respectively).

As a physical space, the Bruhat-Tits tree can play two roles: (1) non-Archimedean string worldsheet, and (2) non-Archimedean analogue of $\mrm{AdS}_2$. In the present paper, we will only be concerned with the tree as string worldsheet. Discussions of the tree as $\mrm{AdS}_2$ can be found in \cite{Gubser:2016guj,Heydeman:2016ldy,Gubser:2016htz,Stoica:2018zmi}. Note, however, that unlike in the Archimedean case, in the \padic world the distinction between string worldsheet and $\mrm{AdS}_2$ is not as clear-cut.

\textbf{Notation:} The boundary $\pd T_p$ of the tree can be identified with $\mathbb{P}^1\lb \mathbb{Q}_p \rb$. We denote vertices in the bulk of the tree $T_p$ by Latin indices $i,j,\dots$, and vertices on the boundary $\pd T_p$ by Greek indices $\alpha,\beta,\dots$. If $i$ and $j$ are neighboring vertices we write $i\sim j$.

\begin{figure}[t]
\input{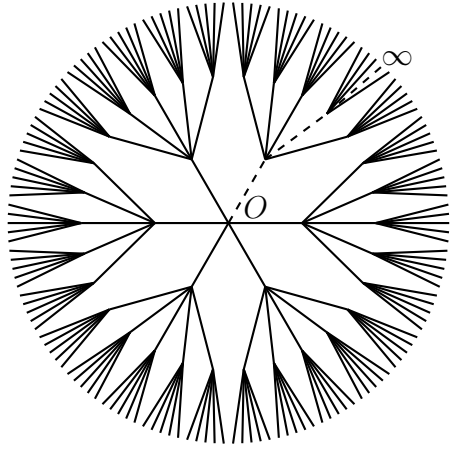}
\caption{Bruhat-Tits tree for $p=5$. We have fixed a center $O$ and the point at infinity $\infty$. These two points determine a half-geodesic, pictured here with dashed line segments. The stabilizer of $\lcb O, \infty \rcb$ is $B\cap \mrm{PGL}\lb2, \mathbb{Z}_p\rb$.}
\label{bttree}
\end{figure}

\subsection{The tree as lattice equivalence classes}

In this subsection we will prove some technical results that will be needed in the rest of the paper. These technical results originate from the well-known fact that the vertices in the Bruhat-Tits tree are lattice equivalence classes for $\mathbb{Q}_p^2$ \cite{casselman,goren}.

\begin{defn}
A lattice $\LL$ for $\mathbb{Q}_p^2$ is
\be
\LL = \mathbb{Z}_p v_1 + \mathbb{Z}_p v_2,
\ee
for two independent vectors $v_{1,2}\in\mathbb{Q}_p$. We will also denote lattice $\LL$ as $\lb v_1,v_2\rb.$
\end{defn}

\begin{defn}
Two lattices $\LL_{1,2}$ are equivalent if $\LL_1 = \alpha \LL_2$, for $\alpha\in\mathbb{Q}_p^\times$. We denote the equivalence class of lattice $\LL$ by $[\LL]$.
\end{defn}

\begin{comm}
\label{commm1}
Two lattice equivalence classes $[\LL_{1,2}]$ are neighbors in the tree if
\be
\LL_1 \supset \LL_2 \supset p \LL_1.
\ee
This definition is reflexive. Equivalently, the $p+1$ neighbors of the lattice spanned by vectors $\lb v_1,v_2\rb$ are lattices $\lb v_1 , p v_2 \rb$ and $\lb p v_1 , i v_1 +v_2 \rb$, with $i\in\{0,\dots,p-1\}$.
\end{comm}

\begin{comm}
A matrix $M\in \mrm{PGL}\lb2, \mathbb{Q}_p\rb$ acts on lattice $\LL=\lb v_1,v_2 \rb$ as
\be
\lb v_1,v_2\rb \to \lb M v_1, M v_2 \rb.
\ee
\end{comm}

In order to describe waves propagating on the tree, we will need to fix a center point $O$ for the tree and a ``point at infinity'' $\infty$ on $\pd T_p$. We take the lattice representations of $O$ and $\infty$ to be
\be
\LL_O = \begin{pmatrix}
	1 & 0\\
	0 & 1\\
\end{pmatrix}, \quad 
\LL_\infty = \begin{pmatrix}
	1 & 0\\
	0 & 0\\
\end{pmatrix}.
\ee 
We will be interested in the stabilizer subgroup $S$ of $\mrm{PGL}\lb2, \mathbb{Q}_p \rb$ which fixes $O$ and $\infty$. This is characterized in Lemma \ref{stablemma} below.

\begin{lemma}
\label{stablemma}
The stabilizer $S$ of $\lcb O,\infty \rcb$ is
\be \label{local symmetry group}
S \lb O,\infty \rb = B \cap \mrm{PGL}\lb2, \mathbb{Z}_p \rb,
\ee
with $B$ the standard Borel subgroup of upper triangular matrices.
\end{lemma} 

\begin{proof}
The stabilizer of $O$ is $\mrm{PGL}\lb2, \mathbb{Z}_p \rb$. The stabilizer of $\infty$ consists of upper-triangular matrices.
\end{proof}
We will also need the following lemma.

\begin{lemma}
\label{lemmalat}
Any lattice equivalence class $[\LL]$ has a representative by an upper-triangular matrix in the standard Borel.
\end{lemma}
\begin{proof}
Suppose a representative of $[\LL]$ is
\be
\begin{pmatrix} 
a & b \\
c & d 
\end{pmatrix}, \quad a,b,c,d \in \mathbb{Q}_p.
\ee
If $v(c) \geq v(d)$, 
\be
\begin{pmatrix} 
a & b \\
c & d 
\end{pmatrix}
\begin{pmatrix} 
1 & 0 \\
-cd^{-1} & 1 
\end{pmatrix} =
\begin{pmatrix} 
a-bcd^{-1} & b \\
0 & d 
\end{pmatrix}
\ee
gives an upper-triangular representative. If $v(c) < v(d)$, the representative is instead given by
\be
\begin{pmatrix} 
a & b \\
c & d 
\end{pmatrix}
\begin{pmatrix} 
-\frac{d}{c} & 1 \\
1 & 0 
\end{pmatrix} =
\begin{pmatrix} 
b-adc^{-1} & a \\
0 & c 
\end{pmatrix}.\qedhere
\ee
\end{proof}

\subsection{Eigenfunctions of the Laplacian}

Because propagation of strings on the worldsheet is intimately connected to the Laplacian, we need to study the Laplacian equation on the Bruhat-Tits tree $T_p$, for functions
\be
\phi: V(T_p) \to \mathbb{C}.
\ee

\begin{defn}
The Hecke operator $\TT$ on the tree acts at vertex $i$ as
\ba
\TT \phi(i) = \sum_{j\sim i} \phi(j).
\ea
\end{defn}
\begin{defn}
The Laplacian operator is
\be
\Lap \coloneqq \TT - p -1,
\ee
acting at vertex $i$ as
\be
\Lap\! \phi (i) = - (p+1) \phi(i) + \sum_{j\sim i} \phi(j).
\ee
\end{defn}

\begin{comm}
The convention of the Laplacian we use here is different from the usual convention in spectral graph theory by a minus sign.
\end{comm}

To describe strings on the worldsheet, we need an eigenbasis for the momentum operator. Such a basis was constructed by Zabrodin \cite{zabrodin,zabrodin2}, as an eigenbasis for the Laplacian operator. We reproduce this construction in Lemma \ref{thm1} below, but in order to do so we will need to recall some more definitions, namely the path overlap function $\delta$ and the propagation bracket $\left\langle\cdot,\cdot\right\rangle$.

\begin{comm}
In order to write down Laplacian eigenfunctions on the Bruhat-Tits tree we also need to fix a ``center vertex'' (or base point) $O$, akin to specifying the origin of the coordinate system in Archimedean $\mrm{AdS}_2$. Fixing this center is deeply connected to renormalization of the plane waves.
\end{comm}

We will come back to the relation between specifying the center and renormalization in Comment \ref{comm3}, after constructing the Laplacian eigenfunctions. Fixing the center vertex $O$ is also necessary in order to construct the state space at $O$; we will come back to this point in Section \ref{statelabel}.

\begin{defn}[Path overlap function]
If $i$, $j$ are vertices in the tree (or on the boundary), then $i\to j$ represents the unique \emph{oriented} path in the tree from $i$ to $j$, and $\ell\lb i,j \rb = \ell\lb j,i \rb$
is the number of edges between $i$ and $j$. Furthermore,
\be
\delta\lb i\to j, k\to l \rb
\ee
is the signed length of the common part of paths $i\to j$ and $k \to l$.
\end{defn}
\begin{defn}[Propagation bracket]
With base point $O$, we introduce
\be
\label{propbrak}
\langle i,j \rangle \coloneqq \delta \lb O\to j, O \to i\rb + \delta\lb i \to j, O \to i \rb.
\ee
\end{defn}

Let's now define the momentum operator $P_\alpha$. This operator acts as translation by one in the direction toward a marked boundary vertex $\alpha$.

\begin{defn}
The momentum operator $P_\alpha$ acts on functions $\phi:V\lb T_p \rb\to \mathbb{C}$ as
\be
P_\alpha \phi(i) = \phi(j),
\ee
where $j$ is the neighbor of $i$ in the direction of increasing propagation bracket $\langle\cdot,\alpha\rangle$.
\end{defn}

The Hecke eigenfunction equation at vertex $j$ is
\be
\label{heckeeigen}
\TT \phi(j) = \lambda \phi(j).
\ee

\begin{lemma}[Zabrodin]
\label{thm1}
There are eigenfunctions $\phi_{\mu,\alpha}$ for the Hecke operator with eigenvalues $\lambda_\mu$, given as follows: 

\ba
\label{eigenfctplane}
\phi_{\mu,\alpha}(i) &=& p^{\mu\langle i,\alpha \rangle}, \\
\label{eq219spec}
\lambda_\mu &=& p^\mu + p^{1-\mu},
\ea
where $\mu\in\mathbb{C}$ is an arbitrary parameter, and $\alpha\in\partial T_p$ is any boundary point.
\end{lemma}

It is easiest to verify Lemma \ref{thm1} using the following Lemma \ref{lemma1}.

\begin{lemma}
\label{lemma1}
Let $\alpha$ be any boundary point. We have
\be
\label{langleeq}
\langle i,\alpha \rangle = \begin{cases}
-\ell(i,O) ,\qquad \qquad \mrm{if\ } \alpha,O,i \mrm{\ are\ on\ the\ same\ geodesic, \ in\ this\ order}\\
\ell(O,i') - \ell(i',i) \quad \ \mrm{if\ not}
\end{cases}.
\ee
Here $i'$ (not necessarily distinct from $i$) is the vertex at which the $\alpha \to O$ and $\alpha \to i$ paths stop coinciding. 
\end{lemma}

The proof of Lemma \ref{lemma1} is immediate. Then, to prove Lemma \ref{thm1}, there are four cases to consider, as in Figure \ref{fig1}. In each case, using expression \eqref{langleeq}, the eigenfunction equation can be checked explicitly.

\begin{figure}[htp]
\centering
\centering
\begin{tikzpicture}[scale=0.7]
	\begin{pgfonlayer}{nodelayer}
		\node [style=simpletoo] (0) at (-5, -0) {};
		\node [style=simpletoo] (1) at (-4, 2) {};
		\node [style=simpletoo] (2) at (-3, -0) {};
		\node [style=simpletoo] (3) at (-1, 4) {};
		\node [style=simpletoo] (4) at (-3, 4) {};
		\node [style=simpletoo] (5) at (-1, -0) {};
		\node [style=simpletoo] (6) at (1, 4) {};
		\node [style=simpletoo] (7) at (0, 2) {};
		\node [style=simpletoo] (8) at (1, -0) {};
		\node [style=simpletoo] (9) at (3, 4) {};
		\node [style=simpletoo] (10) at (2, 2) {};
		\node [style=simpletoo] (11) at (3, -0) {};
		\node [style=none] (12) at (-5, -0.5) {$\infty$};
		\node [style=none] (13) at (-3, -0.5) {$\infty$};
		\node [style=none] (14) at (-1, -0.5) {$\infty$};
		\node [style=none] (15) at (1, -0.5) {$\infty$};
		\node [style=none] (16) at (3, -0.5) {$i$};
		\node [style=none] (17) at (-3, 4.5) {$i$};
		\node [style=none] (18) at (-1, 4.5) {$i=O$};
		\node [style=none] (19) at (1, 4.5) {$O$};
		\node [style=none] (20) at (3, 4.5) {$O$};
		\node [style=none] (21) at (-4.5, 2) {$O$};
		\node [style=none] (22) at (-0.5, 2) {$i$};
		\node [style=none] (23) at (2.5, 2) {$i'$};
	\end{pgfonlayer}
	\begin{pgfonlayer}{edgelayer}
		\draw [style=simple] (0.center) to (1.center);
		\draw [style=simple] (1.center) to (4.center);
		\draw [style=simple] (2.center) to (3.center);
		\draw [style=simple] (5.center) to (7.center);
		\draw [style=simple] (7.center) to (6.center);
		\draw [style=simple] (8.center) to (10.center);
		\draw [style=simple] (10.center) to (9.center);
		\draw [style=simple] (10.center) to (11.center);
	\end{pgfonlayer}
\end{tikzpicture}
\caption{All possible configurations of points $\infty$, $i$, and $O$ for the proof of Lemma~\ref{thm1}.}
\label{fig1}
\end{figure}
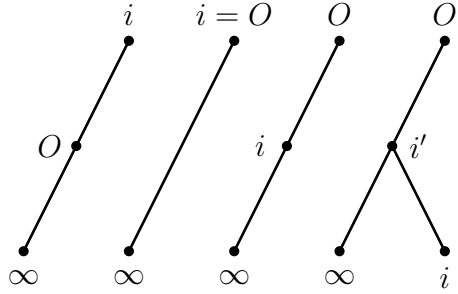

The Laplacian eigenfunctions \eqref{eigenfctplane} are momentum eigenfunctions, as we now explain.

\begin{comm} \label{eigen}
The momentum operator eigenequation is 
\be
P_\alpha \phi_i = \kappa \phi_i = \phi_j,
\ee
with $j$ the neighbor of $i$ in the direction of increasing propagation bracket $\langle\cdot,\alpha\rangle$, and $\kappa\in\mathbb{C}^*$ ($\kappa$ should be nonzero, as $P_\alpha$ is a discrete momentum). The unique solution to this equation that is invariant under the local symmetry group \eqref{local symmetry group} is, up to a multiplicative constant,
\be 
\label{momentumeigenfunction}
\phi_{\mu,\alpha}(i) = p^{\mu \langle i,\alpha\rangle},
\ee
with $\kappa\eqqcolon p^{\mu}$. The Laplacian eigenfunctions \eqref{eigenfctplane} thus are momentum eigenfunctions.
\end{comm}

\begin{comm}
\label{comment6}
The transformation $\mu\to1-\mu$ is a symmetry of the eigenvalues $\lambda_\mu$. Furthermore, the eigenvalues are real iff
\be
\Re \lb \mu\rb = \frac{1}{2} \qquad \mrm{or} \qquad \Im\lb \mu \rb = k\pi/\log{p},\ k\in\mathbb{Z}.
\ee
The different values of $\mu$ of the form $\mu = \Re \lb \mu\rb + i k \pi /\log p$, for fixed real part $\Re \lb\mu\rb$ and any even integer $k$, all correspond to the same eigenfunction of the Laplacian. We will thus mod out by this artificial reparameterization symmetry for even $k$ when computing the two-point functions in Section \ref{pathintsec}.
\end{comm}

\begin{comm}
In this paper we are mostly interested in analyzing one finite place $p$. However, when putting all places together in an adelic structure, we will demand the Archimedean string worldsheet mass squared $m_\infty^2$ to arise from the eigenmodes at places $p$ via the product formula
\be
\label{eqcomm224}
\lb 2\pi \rb^2 m^2_{(\infty)} = \prod_p \frac{1}{p^{-1}\lb 1-p^\mu\rb\lb1-p^{1-\mu}\rb}.
\ee
Here the right-hand side comes from producting the negative of the Bruhat-Tits tree Laplacian spectrum given by $-p-1+\lambda_\mu$, with $\lambda_\mu$ in Eq. \eqref{eq219spec}, and we have divided by a factor of $p$ on the denominator to account for the nontrivial weights (dimensions) of the mass squared: i.e. on the $p$-adic worldsheet, the spatial unit circle has $p$ points, thus has $p$-adic length equal to $p^{-1}$, while the time unit circle has $1$ point with $p$-adic length equal to 1. (Also note that the product of $p$ as defined by analytic continuation, equals $\lb2\pi\rb^2$ on the left-hand side, which is the product of two dimensionful unit circles on the Archimedean side.)

By analytic continuation, (see \eqref{prescription8} and the explanation afterwards about the regularization) the right-hand side of Eq. \eqref{eqcomm224} equals
\be
\zeta_{\lb \infty \rb}\lb \mu\rb \zeta_{\lb \infty \rb}\lb 1-\mu \rb,
\ee
with $\zeta_{\lb \infty \rb}$ the Riemann Zeta function. Thus, in order for the Archimedean mass $m_{(\infty)}^2$ to equal zero, $m_{(\infty)}^2=0$, we must sit at a nontrivial Riemann zero. The trivial Riemann zeros do not appear, due to the regularization in \eqref{prescription8} below, or due to appendix \eqref{appvanishing}. Conversely, if $\mu$ is a Riemann zero on the critical line, Theorem \eqref{thmnorm} implies that it does give rise to a $p$-adic string state with nonzero norm and positive real Laplacian eigenvalue (in our conventions the Laplacian is negative semi-definite), and it does contribute to the $p$-adic string path integral as shown in section \eqref{one loop}. On the other hand, for any nonzero $\mu$ that does not lie on the critical line, a similar calculation as in appendix \eqref{appvanishing} indicates that it does not contribute to the $p$-adic path integral. There are indications that the $\mu=0$ mode may be related to the tachyon. We plan to investigate these issues in detail in a separate paper. 
\end{comm}

\begin{comm}
\label{commm8}
The oscillation modes on the line $\Im\lb \mu \rb = \pi/\log{p}$ will not pair correctly in order to contribute to the $p$-adic string path integrals in Section \ref{pathintsec}, except at the point $\Re\lb\mu\rb=0$. However, at this point the mode $p^{\mu\langle i,\infty\rangle}$ takes alternating values in $\lcb \pm 1\rcb$, and does not have a well-defined boundary value as $i\to\infty$. All other modes have well-defined asymptotic boundary conditions, and in fact the position of the boundary can be defined in terms of this asymptotic behavior.  For the $\Im\lb \mu \rb = i\pi/\log{p}$ mode there is no well-defined boundary, thus we will exclude this mode from discussion in the rest of the paper, i.e. exclude the case when $k$ is odd in Comment \ref{comment6}.
\end{comm}

\begin{comm}
\label{comm3}
Eigenfunctions \eqref{eigenfctplane} are plane waves. This follows from the fact that the sign choices in Eqs. \eqref{langleeq} correspond to the exponent incrementing as the wave propagates on the graph, starting from the boundary point $\alpha$. However, since the graph is infinite, the distance from any point in the bulk to $\alpha$ is divergent. To remove this divergence, the steps in the exponent are instead counted in reference to the base point $O$. Thus, choosing a base point is related to renormalization.
\end{comm}

We are thus led to the following result, that principal series representations of $\mrm{PGL}\lb2,\mathbb{Q}_p\rb$ are plane waves on the Bruhat-Tits tree.

\begin{remark}
Equation \eqref{eigenfctplane}, when $\Re(\mu)=1/2$, has an interpretation as a $p$-adic automorphic form: it is the unique spherical vector (up to a nonzero multiple) of the unitary principal series representation of $\mrm{PGL}(2,\mathbb{Q}_p)$.
\end{remark}

More generally, we expect certain $p$-adic automorphic forms to describe fields on the $p$-adic string worldsheet.

If eigenfunctions \eqref{eigenfctplane} are Bruhat-Tits tree plane waves, they must satisfy orthonormality conditions analogous to the Archimedean ones. This is captured by Definition \ref{defip} and Theorem \ref{thmnorm} below. However, before discussing orthonormality, we need to introduce the regularization scheme that we will employ. This regularization scheme is necessary because many of the sums on the tree that we will encounter are naively divergent.

\begin{defn}
\label{prescription8}
Our regularization prescription for the evaluation of a meromorphic function $f$ at a point $x$ is as follows: expand the function $f$ in a Laurent expansion at point $x$ (which is allowed to be a pole), and pick out the constant term. This regularization and its variants are used often in physics.

We will mostly be interested in applying this prescription to geometric sums:
\begin{enumerate}
\item If $ s >0$, the sum
\be
\label{eqq2222}
\sum_{k=1}^\infty p^{-ks}= \frac{1}{p^s-1}
\ee
is finite, so the constant term is just the value of the sum.
\item If $s<0$, writing
\be
\label{eqq22222}
\sum_{k=1}^\Lambda p^{-ks}= \frac{1}{p^s-1} - \frac{p^{-s\Lambda}}{p^s-1},
\ee
for a large cutoff $\Lambda$, has a pole in $p^{s\Lambda}$ in the second term as $\Lambda\to\infty$, which gets discarded.
\item If $\Re\lb s\rb=0$ and $\Im \lb s \rb \neq 2ki/\log p$, $k\in\mathbb{Z}$, we can write $s=\epsilon + i \Im\lb s\rb$ and take the limit $\epsilon \to 0$ in Eqs. \eqref{eqq2222} or \eqref{eqq22222}, to obtain
\be
\sum_{k=1}^\infty p^{-ks}= \frac{1}{p^s-1}.
\ee
\item If $s=2ki/\log p$, $k\in\mathbb{Z}$, we can expand the right-hand side in Eq. \eqref{eqq2222} as 
\be
\frac{1}{p^s-1} = \frac{1}{s \log (p)}-\frac{1}{2} + \OO\lb s \rb,
\ee
which contains a log singularity in $p^s$ as $s\to 2ki/\log p$. The constant term in this case~is~$-1/2$.
\end{enumerate}
\end{defn}
More details on this method of regularization can be found in \cite{renreb}.

For geometric sums, the regularization scheme above coincides with analytic continuation, for which the geometric sum
\be
\sum_{k=1}^\infty p^{-ks}= \frac{1}{p^s-1} , \quad \Re(s) >0,
\ee
gets analytically extended to the entire complex plane, minus the poles at $s=2ki/\log p$, $k\in\mathbb{Z}$. Furthermore, at these poles the geometric series is the divergent sum $1+1+1+\dots$, and can also be thought of as being regularized as $\zeta(0)=-1/2$, with $\zeta$ the Riemann Zeta function. We are thus led to the following form of our regularization.
\begin{comm}
\label{def6}
The geometric series is regularized by the prescription in Definition \ref{prescription8} (or by analytic continuation) as
\be
\label{reghere}
\sum_{j=1}^\infty p^{-js}= \begin{cases} 
\frac{1}{p^s-1} , \quad s \neq \frac{2\pi ik}{\log p}\\
-\frac{1}{2}, \quad\  s = \frac{2\pi ik}{\log p}
\end{cases},
\ee
with $k\in\mathbb{Z}$.
\end{comm}
Analytic continuation \eqref{reghere} is very much in the spirit of the analytic continuation used in defining special functions, e.g. Gamma, Beta and Zeta functions in the Archimedean case. The connection with special functions is in fact stronger: in the $p$-adic case, the Gel'fand-Graev Gamma function and the Beta function are defined precisely via analytic continuation \eqref{reghere} (see e.g. \cite{GGIP}). Prescription \eqref{reghere} can be used as a regularization scheme in other contexts also, however we will not detail this in the present paper.

We now introduce the inner product for functions defined at the vertices of the tree. As mentioned above, this inner product is in all cases divergent, and needs to be regularized using the prescription in Definition \ref{prescription8} and Eq. \eqref{reghere}.

\begin{defn}
\label{defip}
The inner product of functions $\psi,\phi:V(T_p)\to\mathbb{C}$ is
\be
\label{innerp}
\langle \phi | \psi \rangle \coloneqq \sum_{i\in V(T_p)} \phi^*_i \psi_i,
\ee
with the star denoting conjugation, defined to act on eigenfunctions as
\be
\phi_{\mu,\alpha}^* \coloneqq \phi_{1-\mu,\alpha}.
\ee
\end{defn}

On the critical line this definition of conjugation is the same as complex conjugation. It is also related to time reversal symmetry for a certain notion of time, as we will explain in Section \ref{fixedbvert}.

\begin{thm}
\label{thmnorm}
With regularization \eqref{reghere}, the eigenvalues $\phi_{\mu,\alpha}$ satisfy the normalization condition
\be
\label{hereisnorm0}
\langle\phi_{\mu,\alpha} | \phi_{\nu,\beta} \rangle =  \frac{1}{1-p} \lb \delta_{\mu^*+\nu} + \delta_{\mu^*+\nu-1}\rb \delta_{\alpha,\beta} + \frac{p^{\ell(O,O')}-1}{p-1} \delta_{\mu^*+\nu-1},
\ee
with $O'$ the vertex on the $\alpha \leftrightarrow \beta$ geodesic closest to $O$ ($O'=O$ if $O$ is on the $\alpha \leftrightarrow \beta$ geodesic, or if $\alpha=\beta$), where $\mu,\nu$ are either real, or lie on the critical line.
\end{thm}

\begin{figure}[t]
\input{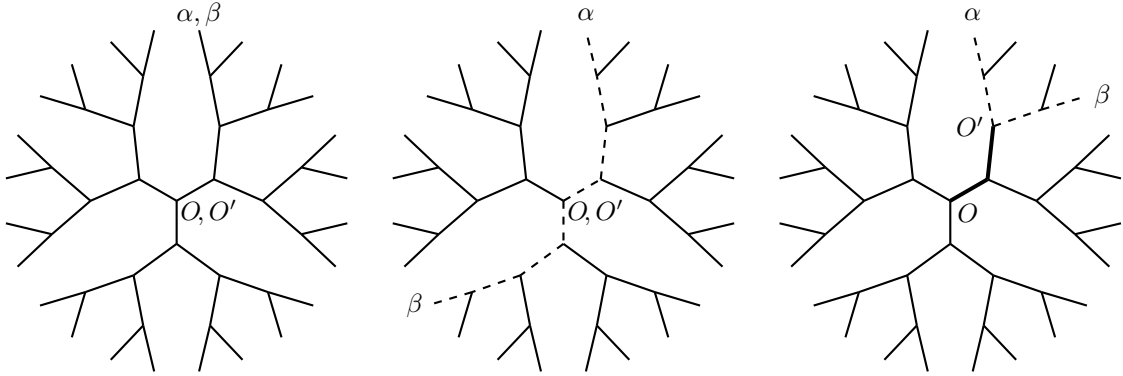}
\caption{The three possible relative positions of points $\alpha$, $\beta$, $O$ and $O'$. The geodesic $\alpha\leftrightarrow\beta$ (if it exists) is represented dashed, and the distance from $O$ to $\alpha\leftrightarrow\beta$ by thickened segments. The last term in Eq. \eqref{hereisnorm0} only appears in the third configuration. On this figure $p=2$ and $\ell\lb O,O' \rb=2$.}
\label{OOOprime}
\end{figure}

In expression \eqref{hereisnorm0} $\delta_\star$ and $\delta_{\star,\star}$ are Kronecker delta functions. All possible relative positions of vertices $\alpha$, $\beta$ and $O$ are shown in Figure \ref{OOOprime}, and the proof of Theorem \ref{thmnorm} is given in Appendix \ref{appA}. Note that $\alpha,\beta\in\pd T_p$, so they can be thought of as points of $\mathbb{P}^1(\mathbb{Q}_p)$. The Kronecker delta function $\delta_{\alpha,\beta}$ in Eq. \eqref{hereisnorm0} is off by a normalization factor from the Dirac delta function of $\mathbb{P}^1(\mathbb{Q}_p)$, however this normalization factor will not be important in the present paper.

\begin{comm}
We will study the effects of the first and second term in Eq. \eqref{hereisnorm0} in the rest of the paper. Interpreting the half geodesic from $\alpha$ to $O$ as a time direction (which we will explain in Section \ref{fixedbvert} below), Eq. \eqref{hereisnorm0} implies that there are $p+1$ independent movers on the worldsheet, each with a different forward time direction, since one can pick $p+1$ boundary points so that the third term in Eq. \eqref{hereisnorm0} does not contribute.
\end{comm}

\subsection{Worldsheet time direction and local symmetry transformations}
\label{fixedbvert}
\label{statelabel}

In the previous section the plane waves were defined with respect to arbitrary reference vertices on the boundary. However, in the rest of the paper we will fix \emph{one} boundary vertex, denoted $\infty$, with respect to which all plane waves are defined. Doing so discards the third term in the plane wave normalization \eqref{hereisnorm0}, and is related to introducing a kind of time direction on the tree, as we will explain shortly. Furthermore, we have required the scalar field configuration on the tree to respect the symmetry given by the stabilizer $S\lb O,\infty\rb$ (see comment \ref{eigen}); this can be thought of as an analog of a local (Lorentz-, or rotation-like) symmetry group on the discrete worldsheet.

To understand how the stabilizer $S\lb O,\infty\rb$ acts, note that by Lemma \ref{lemmalat}, the matrix representatives of the lattice equivalence classes can be taken to be upper triangular. The following Lemma relates the propagation bracket to these representatives.

\begin{lemma}
With center vertex $O$, the propagation bracket \eqref{propbrak} is related to the matrix representatives
\be
\begin{pmatrix}
a & b \\
0 & d
\end{pmatrix}
\ee 
of the lattice equivalence classes by
\be
\label{230label}
\langle i,\infty \rangle = v(d) - v(a) - \lsb v(d_O) - v(a_O) \rsb.
\ee
\end{lemma}
If $O$ has representative the identity matrix the square bracket vanishes.
\begin{proof}
Let $i'$ be the vertex at which the $\alpha \to O$ and $\alpha \to i$ paths stop coinciding, as before. From Comment \ref{commm1}, the matrix representative for a vertex $i$ for which vertices $\infty,O,i$ are on the same geodesic in this order ($i=i'$ in this case) can be written as
\be
\begin{pmatrix}
p^{\ell\lb i',O \rb} a_O & \star  \\
0 & d_O
\end{pmatrix}.
\ee
If instead vertices $\infty,O,i$ are not on one geodesic in this order, the matrix representative can be written as
\be
\begin{pmatrix}
p^{\ell\lb i,i'\rb} a_O & \star \\
0 & p^{\ell\lb i',O \rb} d_O
\end{pmatrix}.\qedhere
\ee
\end{proof}

\begin{lemma}
The bracket $\langle i,\infty\rangle$ is invariant under the stabilizer $S\lb O,\infty \rb$.
\end{lemma}
\begin{proof}
Immediate from Eq. \eqref{230label}.
\end{proof}

The $P_\infty$ momentum eigenfunctions also respect the $S\lb O,\infty\rb$ symmetry.

We now explain the intuition behind interpreting the special half-geodesic $\lb O,\infty \rb$ as a time direction for the scalar degrees of freedom, and the stabilizer $S\lb O,\infty\rb$ as a rotation-like local symmetry group. In the Archimedean case, there is a sharp distinction between the manifold $\MM$ and the local symmetry group acting (on say the tangent space) at a point; this distinction becomes considerably blurred in the case of the Bruhat-Tits tree.

Consider the momentum operator $P_\alpha$ acting in all $p+1$ directions on an eigenfunction $p^{\mu\langle i,\infty\rangle}$ at a vertex $i$. This produces an eigenvalue vector
\be
\lb p^{\mu}, p^{-\mu}, \dots, p^{-\mu} \rb,
\ee
so at each vertex $i$ the direction pointing toward $\infty$ is distinguished. We identify this direction as a time direction. Then the half-geodesic $\lb\infty,O\rb$ acts as a time direction for center $O$. At each vertex $i$, the stabilizer subgroup $S\lb O,\infty \rb$ acts by identifying (i.e. permuting) the $p$ ``spacelike'' neighbors of $i$, thus it acts similarly to a rotation. Note that the action of $S\lb O, \infty\rb$ identifies all vertices $j$ of a fixed bracket value $\langle j,\infty \rangle$ which share a common leg $j'$ on the $\lb O,\infty\rb$ geodesic. 

The symmetries discussed above are tied to the notion of conjugation we have introduced, and they have consequences for the operator-state correspondence on the tree, as we now briefly explain.

\begin{comm} Given a Laplacian eigenvalue $\Lambda_\mu$, there are two momentum eigenvectors, related by $\mu\to1-\mu$. This is analogous to having momentum eigenvectors of different orientations and the same energy in the Archimedean case. The $\mu$ and $1-\mu$ eigenvalues thus correspond to incoming and outgoing states, and conjugation is equivalent to switching the time arrow on the tree.
\end{comm}

{}

\begin{comm}
Imagine that we didn't choose an infinity to start with, and therefore had the full $\mrm{PGL}(2,\mathbb{Q}_p)$ symmetry. Then the only momentum (or Laplacian) eigenfunction respecting that symmetry would be the constant function on the graph.

Physics becomes nontrivial when the symmetry gets broken by a choice of infinity, which as we have seen, gives rise to a time direction of the $p$-adic string, as the path to infinity. Then, given any vertex on the graph, a local state has to transform according to a representation of our leftover symmetry, which plays the role of the local spatial rotation subgroup of the Lorentz group. If one computes states for the trivial representation that is also a momentum eigenstate, one then obtains Eq. \eqref{momentumeigenfunction}.
\end{comm}

\section{Strings from a distance}
\label{stringsfromdistance}

In this section we switch gears to string theory. The worldsheet is the Bruhat-Tits tree $T_p$, and the target space is locally $\mathbb{R}^d$. Following the discussion above, we are interested in functions
\be
\label{Tmap}
X: V(T_p) \to \mathbb{R}^d. 
\ee
Note that since the action \eqref{eq12} is given by a distance, we could also consider the maps $X$ to be functions $X:V(T_p)\to\mathbb{Q}_p^d$, since $\mathbb{Q}_p$ is equipped with a metric. However, we will not do so in the present paper.

\begin{comm}
\label{comm}
The map \eqref{Tmap} has a peculiar feature: the target space $\mathbb{R}^d$ has uncountably many points, however the vertices of the worldsheet $T_p$ are countable. Thus, uncountably many points in the target space have empty preimage in $T_p$. This is in the spirit of \cite{Stoica:2018zmi}.
\end{comm}

\begin{comm} 
We will be comparing the target space position vectors $X_i$, $X_j$, when $i$ and $j$ are neighbors on the worldsheet, so we need them to be ``close.'' Thus, we will be in a situation where $X_i$ and $X_j$ are close in the target space, and separated at order $p$ on the worldsheet.
\end{comm}

As explained in the introduction, we propose that the action takes the form
\be
\label{Sdist}
S = \frac{1}{V} \sum_{\ij\in E(T_p)} \frac{d^2\lb X^a_i,X^b_j \rb}{a^2_\ij},
\ee
with $d$ the geodesic distance on the target manifold and $V$ the degree of each vertex in the Bruhat-Tits tree. The multiplicative factor of $V$ is important to ensure that the action is invariant under certain rescalings, as we will explain in Section \ref{secscale}.

\begin{comm}
The action \eqref{Sdist} is unitless (so one can multiply it with Planck's constant to restore proper units for action).
\end{comm}

\begin{comm}
Action \eqref{Sdist} depends on the distance squared. This square is important to recover the Polyakov action from our action. However, it may also be possible to write other actions that recover the same physics via field redefinitions.
\end{comm}

We now work in a normal coordinate patch in the target space with center denoted by $0$, and parameterize the target space path between points $X_i$ and $X_j$ as 
\be
\label{yis}
y(t) \coloneqq X_i + (X_j-X_i)t, \qquad 0\leq t\leq 1,
\ee
and furthermore define a vector $W$ by
\be
X_j =: X_i + W.
\ee
In this normal coordinate patch we have the metric in terms of coordinates $X_i$
\be\label{metric expansion}
g_{ab}(X_i) = \eta_{ab} - \frac{1}{3} R_{acbd} X_i^c X_i^d + \OO\lb \frac{X^3}{\ell_R^3} \rb, 
\ee
where $R_{acbd}$ is the curvature tensor at the center of the patch. Then the action can be written as
\be
\label{act27}
S[X] = \frac{1}{V} \sum_{\ij \in E(T_p)} \frac{1}{a_\ij^2} \lb \int_{0}^{1} \sqrt{g_{ab} W^a W^b}  dt \rb^2.
\ee

Plugging in \eqref{metric expansion}, and using the antisymmetry of the curvature tensor, action \eqref{act27} becomes
\ba
S[X] &=& \frac{1}{V} \sum_{\ij \in E(\MM)} \frac{1}{a_\ij^2}  \lsb \eta_{ab} \lb  X^a_j - X^a_i \rb \lb X^b_j - X^b_i \rb - \frac{1}{3} R_{acbd} X_i^c X_i^d X_j^a X_j^b \rsb \nn \\
& & + \OO\lb \frac{X^3}{\ell_R^3} \rb .
\label{eqq213}
\ea

Here $\ell_R$ is the length scale over which the target space curvature varies.

\begin{comm}
The size of the region where this coordinate choice is valid is controlled by the length scale $X$ relative to the target space curvature length scale $\ell_R$. I.e. it is valid when the string is mapped to a region whose scale is much less than the curvature length scale.
\end{comm}

\begin{comm}
Analogously, in the usual Archimedean derivation of the target space Einstein equations, the regime of validity of the computation is controlled by $\sqrt{\alpha'}$ and $\ell_R$. In the $p$-adic case, length scale $X$ being smaller than the target space length scale implies $a\ll \ell_R$. Thus, $a$ can be thought of as the $p$-adic analogue of $\sqrt{\alpha'}$.
\end{comm}

\begin{comm}
In the argument above we are assuming the coordinate patch contains the entirety of the worldsheet. We now justify this. First, the mode coefficients $c_{\mu}$ with large absolute value have an exponentially small contribution to the path integral, and can be ignored.
\end{comm}

\section{The path integral}
\label{pathintsec}

\subsection{Scaling invariance}
\label{secscale}

We are interested in considering a $p$-adic string version of the Archimedean worldsheet scaling symmetry. Explicitly, when the target is flat, a way to consider this scaling symmetry is to assign a uniform edge multiplicity (i.e. to consider a number $\lambda>1$ of parallel edges between any two neighbors in the tree). The vertex valence and Laplacian then scale as
\be
\label{thisisscaling}
V \to \lambda V, \quad \Lap \to \lambda \Lap,
\ee 
so that the action derived in the previous section,
\be
\label{action43}
S[X] = \frac{1}{V} \sum_{\ij \in E(\MM)} \frac{1}{a_\ij^2} \eta_{ab} \lb  X^a_j - X^a_i \rb \lb X^b_j - X^b_i \rb 
,
\ee
remains invariant. The scaling of the Laplacian can be understood by writing it in matrix form as
\be
-\Lap = \begin{pmatrix}
p+1 & -1 &  & -1\\
-1 & p+1 &  & \\
  &  & \ddots & \\
  -1 &  &  & p+1\\
\end{pmatrix},
\ee
such that each entry corresponds to a number of neighbors (in the graph with no multiplicity). The Archimedean analog of scaling the edge weight contribution is scaling the worldsheet metric in the Polyakov action, and the analog of scaling $V$ is scaling the volume element.

In general, when the target is not necessarily flat, the scaling symmetry is $V \to \lambda V, \quad d(X_i,X_j)^2 \to \lambda d(X_i,X_j)^2$.

\begin{comm}
There exists a second scaling under which action \eqref{action43} (or \eqref{Sdist}) is invariant:
\be
\label{41}
a_\ij \to \lambda' a_\ij, \quad d(X_i,X_j) \to \lambda' d(X_i,X_j).
\ee
The interpretation of Eq. \eqref{41}  is that scaling the worldsheet edge lengths by $\lambda'$ (keeping multiplicity constant) must be accompanied by a scaling of the target manifold distances in the same way, in order for the action to remain invariant. We expect scalings \eqref{thisisscaling} and \eqref{41} to be equivalent. Scaling \eqref{41} also applies to the gravitational worldsheet actions discussed in \cite{Gubser:2016htz}.
\end{comm}

We would like this scaling symmetry, which is a classical symmetry, to remain a symmetry at the quantum level, to avoid RG running of the target space metric. Said another way, generically there is no scale-invariant way to regularize action \eqref{eqq213}, but we would like the theory to not depend on any parameters that the regularization may introduce.

In particular, we want the target space metric $g_{ab}$ to not RG run, so that the $\langle X^a X^b \rangle$ two-point function must not receive quantum corrections.

\subsection{The two-point function}

From now on we will consider that all edge lengths $a_{\ij}$ are equal to a common edge length $a$. 

The partition function is (see e.g. \cite{zelenovpathintegral} for an introduction to $p$-adic path integrals)
\be
Z[R] = \int DX e^{i S[X,R]}.
\ee
The bare (meaning in the absence of interactions) two-point function is computed~as
\ba
\langle X_i^a X_j^b \rangle_\mrm{bare} &=& \frac{1}{Z[0]} \int DX \exp\lcb \frac{i}{Va^2} \sum_{\kl \in E(\MM)} \lsb \eta_{mn} \lb  X^m_k - X^m_l \rb \lb X^n_k - X^n_l \rb \rsb \rcb X_i^a X_j^b \nn\\
\label{eq54}
\\
\label{2pt1}
&=& \frac{1}{Z[0]} \int DX \exp \lcb - \frac{i}{Va^2} \sum_{k\in V(\MM)}\eta_{mn}  X^m_k \Lap X_k^n\rcb X_i^a X_j^b,
\ea
where in the second equality we have rewritten the sum.

We thus expand $X$ in eigenfunctions of the Laplacian that obey the $S\lb O,\infty \rb= B\cap \mrm{PGL}\lb 2,\mathbb{Z}_p \rb$ symmetry (in particular, keeping the boundary reference vertex $\infty$ fixed), and from now on we suppress notation of the boundary reference index. We have 
\be
\label{alphaexpansion}
X^a_k = \int_{\mu\in D} c_\mu^a \phi_\mu(k).
\ee
Here $D$ is the domain over which eigenvalue parameter $\mu$ runs. We choose $D$ such that the eigenvalues of the Laplacian are real and non-positive (the Laplacian is negative semi-definite in our conventions). Since we have removed the $\Im\lb \mu \rb = \pi/\log p$ line (as explained in Comment \ref{commm8}), $D$ consists of two regions,
\ba
\label{hereD}
D &=&  \lcb 0 \leq \mu \leq 1\rcb \cup \lcb \Re(\mu) = \frac{1}{2}, -\frac{\pi}{\log{p}}\leq \Im(\mu)\leq \frac{\pi}{\log{p}} \rcb \\
  &\eqqcolon& D_\mrm{real} \cup D_\mrm{crit}.
\ea
$D_\mrm{real}$ is on the real line, and $D_\mrm{crit}$ is a period of the eigenvalue function on the critical~line.

Furthermore, as $X$ is a real field, this implies that
\ba
\label{realitycondition}
\overline{c}_{\frac{1}{2}+i t} &=& c_{\frac{1}{2}-it}, \quad t \in \mathbb{R},
\ea
with the overline denoting complex conjugation.

We need to expand the term in the action containing the Laplacian. Using  Eq. \eqref{alphaexpansion}, we have
\be
\sum_{k\in V(\MM)}\eta_{mn}  X^m_k \Lap X_k^n = \sum_{k\in V(\MM)} \int_{\mu,\nu\in D} \Lambda_\nu c_{\mu}^m c_{\nu}^m  \phi_{\mu}(k) \phi_{\nu}(k)
\ee
with $\Lambda_\nu \coloneqq \lambda_\nu-p-1$ the Laplacian eigenvalues.

There are two possible contributions to the integral, according to Eq. \eqref{hereD}, however, it will turn out that only the critical line gives a nonvanishing contribution. Let's now compute this contribution; the vanishing of the $D_\mrm{real}$ contribution, using the regularization in Definition \ref{prescription8}, is shown in Appendix \ref{appvanishing}.

We parameterize the critical line integrals as
\be
\mu = \frac{1}{2} + i t', \quad \nu = \frac{1}{2} + i t,
\ee
so that
\ba
\sum_{k\in V(\MM)}\eta_{mn}  X^m_k \Lap X_k^n \Big|_{D_\mrm{crit}} &=& \frac{1}{1-p} \int_{\mu,\nu\in D_\mrm{crit}} \Lambda_\nu c_{\mu}^m c_{\nu}^m  \delta_{\mu+\nu-1} \\
&=& \frac{1}{1-p}\int_{t,t'\in\mathbb{R}} \Lambda_{\frac{1}{2}+it} c_{\frac{1}{2}+it}^m c_{\frac{1}{2}+it'}^m  \delta_{t+t'} \nn\\
&=& \frac{1}{1-p} \delta t \int_{-\frac{\pi}{\log{p}}}^{\frac{\pi}{\log{p}}} \Lambda_{\frac{1}{2}+it} \left| c^m_{\frac{1}{2}+it} \right|^2 dt, \nn
\ea
with $|\cdot|$ the Lorentzian norm on the worldsheet, $\delta t$ a discretization parameter, and we have restricted the critical line integral to a period of the mode eigenfunction, since the other periods are redundant parameterizations.

Then the two-point function computation \eqref{2pt1} becomes
\ba
\label{513}
\langle X_i^a X_j^b \rangle_\mrm{bare} &=& \frac{1}{Z[0]} \int Dc \exp \lcb - \frac{i\delta t}{Va^2} \frac{1}{1-p} \int_{t} \Lambda_{\frac{1}{2}+it} \left| c^m_{\frac{1}{2}+it} \right|^2 \rcb\times \nn\\
& &\times \int_{t_1,t_2} c_{\frac{1}{2}+it_1}^a c_{\frac{1}{2}+it_2}^b \phi_{\frac{1}{2}+it_1}(i)\phi_{\frac{1}{2}+it_2}(j).
\ea
If $a\neq b$ then the $c^a_{\frac{1}{2}+it}$ integral is again odd and vanishes. 

Eq. \eqref{513} is the usual path-integral prescription for the free two-point function, which can be evaluated using the standard argument. We denote $T\coloneqq\pi/\log{p}$, and discretize the integrals as $t\to \iota \delta t$, with $\iota$ the discrete variable and summation endpoint $I\coloneqq T/\delta t$. Furthermore, we relabel subscripts of the form $1/2 + i t$ on $c$ and $\phi$ by $\iota$. We can then write
\ba
\langle X_i^a X_j^b \rangle_\mrm{bare} &=& \frac{1}{Z[0]} \int Dc \exp \lcb - \frac{i\delta t^2}{Va^2} \frac{1}{1-p} \sum_{\iota = - I}^I \Lambda_{\iota} \left| c^m_{\iota} \right|^2 \rcb \sum_{\iota_1,\iota_2=-I}^I c_{\iota_1}^a c_{\iota_2}^b \phi_{\iota_1}(i)\phi_{\iota_2}(j) \delta t^2 \nn \\
\label{420eq}
&=& \frac{\eta^{ab}}{Z[0]} \sum_{i_1=-I}^I \lb \prod_{\iota} \frac{a}{\delta t} \sqrt{\frac{\pi V \lb p-1 \rb}{ i\Lambda_\iota}} \rb \frac{ V a^2 \lb p-1 \rb}{2 i \Lambda_{\iota_1}} \phi_{\iota_1}(i) \phi_{-\iota_1}(j).
\ea
Note that in obtaining Eq. \eqref{420eq} combinations such as $c^2$ multiplying the exponential vanish by complex integration.

The prefactor cancels against the normalization in Eq. \eqref{420eq} and we can restore the $t$ integral to obtain
\be
\label{eq515}
\langle X_i^a X_j^b \rangle_\mrm{bare} = \frac{iV a^2 }{2 \delta t} \lb 1-p\rb \eta^{ab} \int_t \frac{1}{\Lambda_{\frac{1}{2}+it}} \phi_{\frac{1}{2}+i t}(i) \phi_{\frac{1}{2}-i t}(j).
\ee

In integrating over the Fourier configurations we used a Kronecker delta instead of a Dirac delta function, at the expense of introducing a factor of $\delta t$. Here $1/\delta t$ plays the role of the string tension, so the combination $a^2/\delta t$ that appears in Eq. \eqref{eq515} is natural.

Expression \eqref{eq515} is $i\leftrightarrow j$ exchange symmetric, since $\Lambda_{\frac{1}{2}+i t}=\Lambda_{\frac{1}{2}-i t}$. Eq. \eqref{eq515} can be thought of as the analogue of the momentum space representation of a bare propagator.

\subsection{One-loop correction to the two-point function}
\label{one loop}
We are interested in computing the object
\ba
\label{heres523}
\langle X_i^a X_j^b \rangle_\mrm{1-loop} = \frac{1}{Z[0]} \int DX \exp \lcb - \frac{i}{a^2V} \sum_{k\in V(\MM)}  X^{m'}_k \Lap X_k^{m'}\rcb \times\\
\times \sum_{k_{1}\in V(T_p)} \sum_{k_2\sim k_1} \lb - \frac{1}{6a^2V} R_{mpnq} X_{k_1}^m X_{k_1}^n X_{k_2}^p X_{k_2}^q \rb X_i^a X_j^b, \nn
\ea
where $k_2\sim k_1$ means that the second sum should be over all the neighbors of $k_1$, and the extra factor of $2$ in the denominator accounts for double-counting the edges.

The first step is to evaluate the $k_2$ sum, by expanding in Laplacian eigenfunctions. We have
\ba
\sum_{k_2\sim k_1} X_{k_2}^p X_{k_2}^q  &=& \int_{\lambda,\rho\in D} \sum_{k_2\sim k_1} c_\lambda^p c_\rho^q \phi_{\lambda}\lb k_2 \rb \phi_{\rho}\lb k_2 \rb \\
&=& \int_{\lambda,\rho\in D} \lb p^{1-\lambda-\rho} + p^{\lambda+\rho} \rb c_\lambda^p c_\rho^q \phi_{\lambda}\lb k_1 \rb \phi_{\rho}\lb k_1 \rb,
\ea
where we have used explicit expression \eqref{eigenfctplane} for the eigenfunctions. Then, expressing Eq. \eqref{heres523} in momentum space, we have
\ba
& &\langle X_i^a X_j^b \rangle_\mrm{1-loop} = \frac{1}{Z[0]} \int Dc \exp \lcb - \frac{i}{a^2V} \sum_{k'\in V(T_p)} \int_{\mu_{1,2}} \Lambda_{\mu_2} c^l_{\mu_1}c^l_{\mu_2} \phi_{\mu_1+\mu_2}(k') \rcb \times\nn\\
& & \times \sum_{k\in V(T_p)} \int_{\mu,...,\rho} \frac{-1}{6a^2V} R_{mpnq} \lb p^{1-\lambda-\rho} + p^{\lambda+\rho} \rb c_\mu^m c_\nu^n c_\lambda^p c_\rho^q \phi_{\mu+\nu+\lambda+\rho}\lb k \rb \times \nn\\
& &\times\int_{\mu_{3,4}} c^a_{\mu_3}c^b_{\mu_4} \phi_{\mu_3}(i) \phi_{\mu_4}(j),
\ea
We now compute the $k$ sum; using results in Appendix \ref{appA}, we have
\be
\sum_{k\in V(T_p) } \phi_{\sigma}\lb k \rb =\frac{1}{1-p} \lb \delta_{\sigma} + \delta_{1-\sigma} \rb,
\ee
so that
\ba
& &\langle X_i^a X_j^b \rangle_\mrm{1-loop} = \frac{1}{Z[0]} \int Dc \exp \lcb \frac{i}{a^2V\lb p-1\rb} \int_{\mu_{1,2}} \Lambda_{\mu_2} c^l_{\mu_1}c^l_{\mu_2} \lb \delta_{\mu_1+\mu_2} + \delta_{1-\mu_1-\mu_2} \rb \rcb \times\nn\\
& & \times \int_{\mu,...,\rho} \frac{1}{6\lb p-1\rb a^2V} R_{mpnq} \lb p^{1-\lambda-\rho} + p^{\lambda+\rho} \rb c_\mu^m c_\nu^n c_\lambda^p c_\rho^q \times \nn\\
\label{eq426oneloop}
& &\times \lb  \delta_{\mu+\nu+\lambda+\rho} + \delta_{1-\mu-\nu-\lambda-\rho} \rb \int_{\mu_{3,4}} c^a_{\mu_3}c^b_{\mu_4} \phi_{\mu_3}(i) \phi_{\mu_4}(j).
\ea

Let's now consider the Wick contractions. A priori, every parameter $\mu,\dots\in D$ can either be on the critical line, or in the $0\leq \mu \leq 1$ interval. However, it turns out the Greek indices pair in the path integral such as to give a possibly nonzero contribution only in the following combinations (here the notation means each entry in a bracket on the left and right sides is used precisely once, and the brackets are \emph{unordered}):
\begin{enumerate}
\item $\{\mu_3,\mu_4\}=\{0,0\}$, $\{\mu,\nu,\rho,\lambda\}=\{0,0,0,0\}$.

In this case the contribution is proportional to 
\ba
& &\int Dc\exp\lcb \frac{i}{a^2V(p-1)} \Lambda_0 \lb c_0^l \rb^2 \rcb \times\\
& & \times R_{mpnq} c_0^m c_0^n c_0^p c_0^q c^a_{0}c^b_{0} \phi_{0}(i) \phi_{0}(j),
\ea
which vanishes by antisymmetry of the Riemann tensor.
\item \label{case2} $\{\mu_3,\mu_4\}=\{0,0\}$, $\{\mu,\nu,\rho,\lambda\}=\{0,0,\frac{1}{2}+it,\frac{1}{2}-it\}$.

In this case, remembering reality condition \eqref{realitycondition},  the contribution is proportional~to
\ba
& &\int Dc \exp \lcb \frac{i}{a^2V\lb p-1\rb} \lb \Lambda_0 \lb c^l_0 \rb^2 + \int_{t} \Lambda_{\frac{1}{2}+it} \left| c^l_{\frac{1}{2}+it} \right|^2 \rb \rcb \times\nn\\
& & \times R_{mpnq} c_0^m c_0^n c_{\frac{1}{2}+it}^p c_{\frac{1}{2}-it}^q c^a_{0}c^b_{0} \phi_{0}(i) \phi_{0}(j).
\ea
\item \label{case3} $\{\mu_3,\mu_4\}=\{\frac{1}{2}+it,\frac{1}{2}-it\}$, $\{\mu,\nu,\rho,\lambda\}=\{0,0,\frac{1}{2}+it,\frac{1}{2}-it\}$.

In this case the contribution is proportional to
\ba
& &\int Dc \exp \lcb \frac{i}{a^2V\lb p-1\rb} \lb \Lambda_0 \lb c^l_0 \rb^2 + \int_{t} \Lambda_{\frac{1}{2}+it} \left| c^l_{\frac{1}{2}+it} \right|^2 \rb \rcb \times\nn\\
& & \times R_{mpnq} c_{\frac{1}{2}+it}^m c_{\frac{1}{2}-it}^n c_{0}^p c_{0}^q c^a_{\frac{1}{2}+it}c^b_{\frac{1}{2}-it} \phi_{\frac{1}{2}+it}(i) \phi_{\frac{1}{2}-it}(j).
\ea
\end{enumerate}

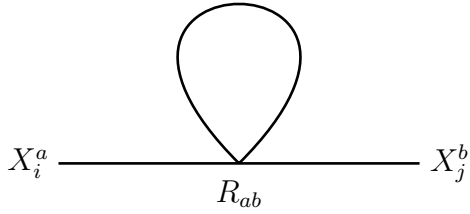
\begin{figure}[htp]
\centering
\centering
\begin{tikzpicture}[scale=0.8]
	\begin{pgfonlayer}{nodelayer}
		\node [style=none] (0) at (-3, -0) {};
		\node [style=none] (1) at (3, -0) {};
		\node [style=none] (2) at (0, -0) {};
                \node [style=none] (3) at (0, -0.5) {$R_{ab}$};
		\node [style=none] (4) at (-3.5, -0) {$X^a_i$};
		\node [style=none] (5) at (3.5, -0) {$X^b_j$};
	\end{pgfonlayer}
	\begin{pgfonlayer}{edgelayer}
		\draw [style=simple] (0.center) to (2.center);
		\draw [style=simple] (2.center) to (1.center);
		\draw [style=simple, in=135, out=45, loop,distance=5cm,fill=white] (2.center) to (2.center);
	\end{pgfonlayer}
\end{tikzpicture}
\caption{One-loop correction to the two-point function.}
\label{figbubble}
\end{figure}

Contributions \ref{case2}, \ref{case3} can be summarized as follows. In order to have a connected diagram with nonzero contribution, either $a,b$ are paired with $m,n$ (and $p$, $q$ with each other), or $a,b$ are paired with $p,q$ (and $m$, $n$ with each other), i.e. contractions of the type
\be
\contraction[2ex]{R_{mpnq}}{c^m}{c^nc^p c^q }{c^a} \contraction[3ex]{R_{mpnq}c^mc^n}{c^p}{}{c^q} \contraction{R_{mpnq}c^m}{c^n}{c^p c^q c^a}{c^b}  R_{mpnq} c^m c^n c^p c^q c^ac^b,
\ee
corresponding to the bubble diagram in Figure \ref{figbubble}. The diagram thus is proportional to the Ricci tensor $R_{ab}$.

For the case when $\mu_3=\mu_4=0$ (case \ref{case2} above), this contribution is naively divergent, and must be regularized by the introduction of an IR regulator $M^2$ in the bare propagator. This is because we are integrating over the $\mu=0$ zero mode on the real line (note that $\mu=0$ is paired with the pole in the Zeta function at $\mu=1$ by the functional equation). 

The leading contribution comes from case \ref{case2}, and is of order $1/M^4$, where $M^2$ is the IR regulator (the bare two-point function for the zero mode is just $1/M^2$). This divergence cannot be removed by our regularization scheme, since it corresponds to a zero mode, so a small variation in $\mu$ at $\mu=0$ keeps the action unchanged by our normalization computation.

In order for the one-loop correction to the two-point function to vanish (and for the target space metric to not RG flow), we must thus have $R_{ab}=0$. This gives the target space Einstein equations. Conversely, setting $R_{ab}=0$ also cancels the subleading divergences in $M$, as well as any finite corrections to the two-point function.

\begin{comm}
In case \ref{case3} the contribution comes from pairing $\mu_{3,4}$ with $1/2\pm i t$, and the zero modes with each other, for an overall divergence of $1/M^2$. Removing this divergence leads to the same result, that the Ricci curvature vanishes.
\end{comm}

\begin{comm}
It would be interesting to explore if the machinery developed in this paper for the computation of tree level and one-loop diagrams could also be used to compute Archimedean diagrams. It seems that the structure of the loop diagrams in the $p$-adic bosonic string is considerably simpler than the Archimedean one, due to normalization \eqref{hereisnorm0}. In particular, the divergences appear to come from the IR zero mode, and the momentum loops are all finite period integrals.
\end{comm}

\section*{Acknowledgments}

We would like to thank Steven S. Gubser, Harsha Hampapura, Matthew Headrick, Narges Iraji, Bong Lian, Maria Monica Nastasescu, Omer Offen, Howard J. Schnitzer, Zhiwei Yun and Dingxin Zhang for useful discussions. The work of A.H. was supported in part by a grant from the Brandeis University Provost Office, and by the Simons Foundation in Homological Mirror Symmetry. The work of B.S. was supported in part by the U.S. Department of Energy under grant DE-SC-0009987, by the Simons Foundation through the It from Qubit Simons Collaboration on Quantum Fields, Gravity and Information, and by a grant from the Brandeis University Provost Office. The work of S.-T. Yau was supported in part by the Simons Foundation in Homological Mirror Symmetry. B.S. would like to thank the Harvard Center for Mathematical Sciences and Applications for hospitality. 

\appendix

\section{Laplacian eigenfunction normalization}
\label{appA}

In this Appendix we compute the Laplacian eigenfunction normalization and prove Theorem \ref{thmnorm}.

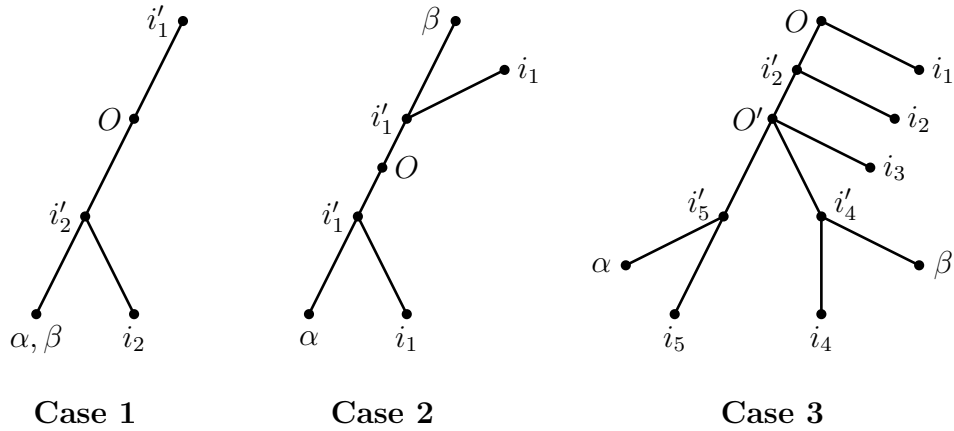
\begin{figure}[h]
\centering
\begin{tikzpicture}[scale=0.65]
 
\tikzset{VertexStyle/.style = {shape = circle,fill = white, minimum size = 0pt,inner sep = 0pt, opacity = 0, text opacity = 1}}
\node [style=simpletoo] (0) at (-2, 2) {};
\node [style=simpletoo] (1) at (-3, -0) {};
\node [style=simpletoo] (2) at (-4, -2) {};
\node [style=simpletoo] (3) at (-2, -2) {};
\node [style=simpletoo] (4) at (-1, 4) {};
\node [style=none] (5) at (-1.5, 4) {$i_1'$};
\node [style=none] (6) at (-2.5, 2) {$O$};
\node [style=none] (7) at (-3.5, -0) {$i_2'$};
\node [style=none] (8) at (-4, -2.5) {$\alpha,\beta$};
\node [style=none] (9) at (-2, -2.5) {$i_2$};
\node [style=none] (24) at (-3, -4) {\textbf{Case 1}};
\draw [style=simple] (1.center) to (2.center);
\draw [style=simple] (1.center) to (3.center);
\draw [style=simple] (1.center) to (0.center);
\draw [style=simple] (0.center) to (4.center);

\end{tikzpicture}
\hspace{1.2cm}
\begin{tikzpicture}[scale=0.65]
\tikzset{VertexStyle/.style = {shape = circle,fill = white, minimum size = 0pt,inner sep = 0pt, opacity = 0, text opacity = 1}}

\node [style=simpletoo] (0) at (-2, 2) {};
\node [style=simpletoo] (1) at (-3, -0) {};
\node [style=simpletoo] (2) at (-4, -2) {};
\node [style=simpletoo] (3) at (-2, -2) {};
\node [style=simpletoo] (4) at (-1, 4) {};
\node [style=none] (5) at (-1.5, 4) {$\beta$};
\node [style=none] (6) at (-2.5, 2) {$i_1'$};
\node [style=none] (7) at (-3.5, -0) {$i_1'$};
\node [style=none] (8) at (-4, -2.5) {$\alpha$};
\node [style=none] (9) at (-2, -2.5) {$i_1$};
\node [style=simpletoo] (10) at (-2.5, 1) {};
\node [style=simpletoo] (11) at (0, 3) {};
\node [style=none] (12) at (-2, 1) {$O$};
\node [style=none] (13) at (0.5, 3) {$i_1$};
\node [style=none] (24) at (-2.5, -4) {\textbf{Case 2}};
\draw [style=simple] (1.center) to (2.center);
\draw [style=simple] (1.center) to (3.center);
\draw [style=simple] (0.center) to (4.center);
\draw [style=simple] (0.center) to (11.center);
\draw [style=simple] (0.center) to (10.center);
\draw [style=simple] (10.center) to (1.center);
  
\end{tikzpicture}
\hspace{0.4cm}
\begin{tikzpicture}[scale=0.65]
\tikzset{VertexStyle/.style = {shape = circle,fill = white, minimum size = 0pt,inner sep = 0pt, opacity = 0, text opacity = 1}}

\node [style=simpletoo] (0) at (-2, 2) {};
\node [style=simpletoo] (1) at (-3, -0) {};
\node [style=simpletoo] (2) at (-5, -1) {};
\node [style=simpletoo] (3) at (-4, -2) {};
\node [style=simpletoo] (4) at (-1, 4) {};
\node [style=none] (5) at (-1.5, 4) {$O$};
\node [style=none] (6) at (-2.5, 2) {$O'$};
\node [style=none] (7) at (-3.5, 0.25) {$i_5'$};
\node [style=none] (8) at (-5.5, -1) {$\alpha$};
\node [style=none] (9) at (-4, -2.5) {$i_5$};
\node [style=simpletoo] (10) at (0, 1) {};
\node [style=none] (11) at (0.5, 1) {$i_3$};
\node [style=simpletoo] (12) at (-1, -0) {};
\node [style=simpletoo] (13) at (-1, -2) {};
\node [style=simpletoo] (14) at (1, -1) {};
\node [style=none] (15) at (-0.5, 0.25) {$i_4'$};
\node [style=none] (16) at (1.5, -1) {$\beta$};
\node [style=none] (17) at (-1, -2.5) {$i_4$};
\node [style=simpletoo] (18) at (-1.5, 3) {};
\node [style=none] (19) at (-2, 3) {$i_2'$};
\node [style=simpletoo] (20) at (0.5, 2) {};
\node [style=none] (21) at (1, 2) {$i_2$};
\node [style=simpletoo] (22) at (1, 3) {};
\node [style=none] (23) at (1.5, 3) {$i_1$};
\node [style=none] (24) at (-2, -4) {\textbf{Case 3}};
\draw [style=simple] (1.center) to (2.center);
\draw [style=simple] (1.center) to (3.center);
\draw [style=simple] (0.center) to (10.center);
\draw [style=simple] (12.center) to (13.center);
\draw [style=simple] (12.center) to (14.center);
\draw [style=simple] (12.center) to (0.center);
\draw [style=simple] (0.center) to (1.center);
\draw [style=simple] (0.center) to (18.center);
\draw [style=simple] (18.center) to (4.center);
\draw [style=simple] (18.center) to (20.center);
\draw [style=simple] (22.center) to (4.center);
  
\end{tikzpicture}
\caption{The configurations of vertices $\alpha$, $\beta$, and $O$ relevant for the proof of Theorem~\ref{thmnorm}.}
\label{fig2}
\end{figure}

\begin{proof}[Proof (Theorem \ref{thmnorm}).]
The normalization condition we want to check is
\be
\label{recoverthis}
\langle\phi_{\mu,\alpha} | \phi_{\nu,\beta} \rangle = \frac{1}{1-p}\lb \delta_{\mu^*+\nu} + \delta_{\mu^*+\nu-1}\rb \delta_{\alpha,\beta} + \frac{p^{\ell(O,O')}-1}{p-1} \delta_{\mu^*+\nu-1}.
\ee
Pick two eigenfunctions $\phi_{\mu,\alpha}$, $\phi_{\nu,\beta}$. The eigenfunction inner product is
\be
\langle \phi_{\mu,\alpha} | \phi_{\nu,\beta} \rangle = \sum_{i\in V(T)} p^{\mu^*\langle i,\alpha \rangle + \nu\langle i,\beta \rangle}.
\ee
There are three cases, as in Figure \ref{fig2}, and we use Lemma \ref{lemma1}.\\
\textbf{Case 1.}
In this case $\alpha=\beta$. We have
\ba
\lb \mu^* + \nu\rb \langle i_1,\alpha \rangle &=& - \lb \mu^* + \nu\rb \ell\lb O,i_1 \rb, \\
\lb \mu^* + \nu\rb \langle i_2,\alpha \rangle &=& \lb \mu^* + \nu\rb \lsb \ell \lb O,i_2' \rb - \ell\lb i_2,i_2' \rb  \rsb,
\ea
so that
\ba
\langle \phi_{\mu,\alpha} | \phi_{\nu,\beta} \rangle &=& \sum_{\ell(O,i_1)=0}^\infty p^{\lb 1-\mu^*-\nu\rb \ell \lb O,i_1 \rb } + \lb \sum_{\ell(O,i_2')=1}^\infty p^{\lb\mu^*+\nu\rb\ell(O,i_2')} \rb\times\nn\\
& &\times\lsb 1 + \frac{p-1}{p} \sum_{\ell\lb i_2,i_2'\rb=1}^\infty  p^{\lb 1-\mu^*-\nu\rb \ell\lb i_2,i_2'\rb} \rsb.
\ea
There are two geometric sums that we need to consider,
\ba
\sum_{i=0}^\infty p^{(1-\mu^*-\nu)i} &=& 1 + \frac{1}{p^{\mu^*+ \nu -1}-1}, \qquad \Re\lb \mu^*+ \nu  \rb > 1, \\
\sum_{i=1}^\infty p^{\lb\mu^*+\nu\rb i} &=&  -1 + \frac{1}{1-p^{\mu^*+\nu }}, \qquad \Re\lb \mu^* + \nu  \rb <0.
\ea

Next, since $\mu,\nu$ are either real, or lie on the critical line, we need only consider the poles of the right hand side of the above two equations at $\mu^*+ \nu=1$ and $\mu^*+ \nu=0$. 

Note that the ranges over which these two series converge do not overlap. However, as explained in Definition \ref{def6}, we introduce analytically continued values as
\ba
\sum_{i=0}^\infty p^{(1-\mu^*-\nu)i} &=& 1 + \frac{1}{p^{\mu^*+ \nu -1}-1}, \qquad  \mu^*+ \nu  \neq 1, \\
\sum_{i=1}^\infty p^{\lb\mu^*+\nu\rb i} &=&  -1 + \frac{1}{1-p^{\mu^*+\nu }}, \qquad  \mu^* + \nu \neq 0.
\ea

Furthermore, we have
\ba
\sum_{i=0}^\infty p^{(1-\mu^*-\nu)i} &=& 1 + \sum_{i=1}^\infty 1 \eqqcolon \zeta_0, \qquad \mu^*+ \nu  = 1,\\
\sum_{i=1}^\infty p^{\lb\mu^*+\nu\rb i} &=& \sum_{i=1}^\infty 1 = \zeta_0, \qquad \mu^*+ \nu  = 0,
\ea
where quantity $\zeta_0$ is formally divergent. This is just $1+1+1+1+\dots$, and can be regularized by 
\be
\label{zetareg}
\sum_{i=1}^\infty 1^i = \zeta\lb0\rb=-1/2,
\ee
however this regularization will not be needed for the first case. This is because when putting everything together $\zeta_0$ cancels, as do the terms not proportional to delta functions, and we obtain
\be
\langle \phi_{\mu,\alpha} | \phi_{\nu,\beta} \rangle = \frac{1}{1-p} \lb \delta_{\mu^*+\nu} + \delta_{\mu^*+\nu-1}\rb,
\ee
with $\delta_\star$ the Kronecker delta function.

\textbf{Case 2.} We have
\ba
\langle i_1,\alpha \rangle &=& - \ell\lb O,i_1 \rb, \\
\langle i_1,\beta\rangle &=& \ell\lb O,i_1' \rb - \ell\lb i_1,i_1' \rb, \\
\langle i_2,\alpha\rangle &=& \langle i_2,\beta\rangle = -\ell\lb O,i_2 \rb,\\
\langle i_3,\beta \rangle &=& - \ell\lb O,i_3 \rb, \\
\langle i_3,\beta\rangle &=& \ell\lb O,i_3' \rb - \ell\lb i_3,i_3' \rb.
\ea
The normalization then is
\ba
\langle \phi_{\mu,\alpha} | \phi_{\nu,\beta} \rangle &=& \lb \sum_{\ell\lb O,i_1' \rb=1}^\infty p^{\lb-\mu^*+\nu\rb\ell\lb O,i_1' \rb} \rb \lb 1 +\frac{p-1}{p} \sum_{\ell\lb i_1,i_1' \rb=1}^\infty p^{\lb1-\mu^*-\nu\rb\ell\lb i_1,i_1' \rb} \rb \nn \\
& & + \lb 1 +\frac{p-1}{p} \sum_{\ell\lb O,i_2 \rb=1}^\infty p^{\lb 1-\mu^*-\nu\rb\ell\lb O,i_2 \rb} \rb + \\
& &\lb \sum_{\ell\lb O,i_3' \rb=1}^\infty p^{\lb\mu^*-\nu\rb\ell\lb O,i_3' \rb} \rb \lb 1 +\frac{p-1}{p} \sum_{\ell\lb i_3,i_3' \rb=1}^\infty p^{\lb1-\mu^*-\nu\rb\ell\lb i_3,i_3' \rb} \rb. \nn
\ea
We need to regularize these geometric sums; as above, we use
\be
\label{geomreg}
\sum_{i=0}^\infty p^{-\alpha i} = 1 - \frac{1}{1-p^{-\alpha}}, \quad \alpha\neq 0.
\ee
There are now several subcases to consider:
\begin{itemize}
\item If $\mu^*\neq \nu$, $\mu^*+\nu\neq 1$, direct computation gives $\langle\phi_{\mu,\alpha} | \phi_{\nu,\beta} \rangle=0$.
\item If $\mu^*=\nu$, $\mu^*+\nu\neq 1$, we have
\be
\langle\phi_{\mu,\alpha} | \phi_{\nu,\beta} \rangle = \frac{(2 \zeta_0+1) \left(p^{2 \nu}-1\right)}{p^{2 \nu }-p},
\ee
which vanishes with Zeta regularization \eqref{zetareg}.
\item If $\mu^*\neq \nu$, $\mu^*+\nu=1$, by direct computation $\langle\phi_{\mu,\alpha} | \phi_{\nu,\beta} \rangle=0$.
\item If $\mu^*=\nu=1/2$, then 
\be
\langle\phi_{\mu,\alpha} | \phi_{\nu,\beta} \rangle = \frac{(2 \zeta_0+1) \lsb\zeta_0 (p-1)+p\rsb}{p},
\ee
which again vanishes with Zeta regularization \eqref{zetareg}.
\end{itemize}
Putting everything together, the answer in case 2 is
\be
\langle \phi_{\mu,\alpha} | \phi_{\nu,\beta} \rangle = 0
\ee
for all values of $\mu$ and $\nu$.

\textbf{Case 3.}
We use Eq. \eqref{geomreg} for regularization. Assume there are $\mathcal{N}$ vertices between $O$ and $O'$ (distinct from $O$ and $O'$), so that $\ell(O,O')=\NN+1$. The sums are
\ba
\langle\phi_{\mu,\alpha} | \phi_{\nu,\beta} \rangle &=& \sum_{\ell\lb O,i_1 \rb=0}^\infty p^{\lb 1-\mu^*-\nu \rb \ell\lb O,i_1 \rb } \nn \\
& & + \sum_{\ell\lb O,i_2' \rb=1}^\NN p^{\lb \mu^* + \nu \rb \ell\lb O, i_2' \rb } \lsb 1 + \frac{p-1}{p} \sum_{\ell\lb i_2,i_2' \rb=1}^\infty p^{ \lb 1 -\mu^* - \nu \rb \ell\lb i_2,i_2' \rb } \rsb \nn \\
& & + p ^{(\mu^*+\nu)\ell\lb O,O' \rb}\lsb 1 + \frac{p-2}{p} \sum_{\ell\lb O',i_3 \rb=1}^\infty p^{ \lb 1 -\mu^* - \nu \rb \ell\lb O',i_3 \rb } \rsb  \\
& & + p^{\lb \mu^*+\nu  \rb \ell\lb O,O' \rb} \sum_{\ell\lb O',i_5' \rb = 1}^\infty \lb p^{\lb\mu^*-\nu\rb\ell\lb O',i_5' \rb} + p^{\lb\nu-\mu^*\rb\ell\lb O',i_5' \rb} \rb \times \nn\\
& &\times \lsb 1 + \frac{p-1}{p} \sum_{\ell\lb i_5,i_5' \rb=1}^\infty p^{\lb 1-\mu^*-\nu\rb \ell\lb i_5,i_5' \rb } \rsb.\nn 
\ea

There are now several subcases to consider:
\begin{itemize}
\item If $\mu^*+\nu\neq 1$, $\mu^* + \nu \neq 0$, $\mu^*\neq \nu$, direct computation gives $\langle\phi_{\mu,\alpha} | \phi_{\nu,\beta} \rangle = 0$.
\item If $\mu^*+\nu =1$, $\mu^*\neq \nu$, we obtain
\be
\langle\phi_{\mu,\alpha} | \phi_{\nu,\beta} \rangle = \frac{p^{\NN+1}-1}{p-1}.
\ee
\item If $\mu^*+\nu =1$, $\mu^*= \nu=1/2$, we obtain, using Zeta function regularization \eqref{zetareg},
\be
\langle\phi_{\mu,\alpha} | \phi_{\nu,\beta} \rangle = \frac{p^{\NN+1}-1}{p-1}.
\ee
\item If $\mu^*+\nu=0$ (including $\mu^*=\nu=0)$, we obtain $\langle\phi_{\mu,\alpha} | \phi_{\nu,\beta} \rangle = 0$ by direct computation.
\item If $\mu^*=\nu$, $\mu^*,\nu\neq0,1/2$, then
\be
\langle\phi_{\mu,\alpha} | \phi_{\nu,\beta} \rangle = \frac{\left(p^{\mu^*+\nu}-1\right) p^{(\NN+1) (\mu^*+\nu)}(2 \zeta_0+1)}{p^{\mu^*+\nu }-p},
\ee
which vanishes using Eq. \eqref{zetareg}.
\end{itemize}

\textbf{Summary.} Putting Cases 1, 2, and 3 together, we have obtained
\be
\label{hereisnorm}
\langle\phi_{\mu,\alpha} | \phi_{\nu,\beta} \rangle = \frac{1}{1-p} \lb \delta_{\mu^*+\nu} + \delta_{\mu^*+\nu-1}\rb \delta_{\alpha,\beta} + \frac{p^{\ell(O,O')}-1}{p-1} \delta_{\mu^*+\nu-1},
\ee
which recovers Eq. \eqref{recoverthis}.
\end{proof}

\section{Vanishing of the real line contributions}
\label{appvanishing}

In this appendix we show that the real line contributions to the two-point function path integral vanish, if one uses the regularization in Definition \ref{prescription8}. We are interested in the quantity
\ba
\label{eqisthisvanishing}
\langle X_i^a X_j^b \rangle_{D_\mrm{real}} &=&  \frac{1}{Z[0]} \int Dc \exp \lcb - \frac{i}{Va^2} \int_{\mu,\nu\in D_2}  \sum_{k\in V(\MM)} \Lambda_\nu c^m_\mu c^m_\nu \phi_\mu(k) \phi_\nu(k) \rcb \times\nn\\
& & \times\int_{\mu_{1,2}\in D_\mrm{real}} c^a_{\mu_1} c^b_{\mu_2} \phi_{\mu_1}(i) \phi_{\mu_2}(j).
\ea
We compute the $k$ sum first; using results from Appendix \ref{appA}, we have (note we are using $0\leq \mu,\nu\leq 1$ to remove one of the delta functions)
\ba
\langle X_i^a X_j^b \rangle_{D_\mrm{real}} &=&  \frac{1}{Z[0]} \int Dc \exp \lcb - \frac{i\delta t}{Va^2} \frac{1}{1-p} \int_{\nu\in D_2} \Lambda_\nu c^m_{1-\nu} c^m_\nu \rcb \times\nn\\
& & \times\int_{\mu_{1,2}\in D_\mrm{real}} c^a_{\mu_1} c^b_{\mu_2} \phi_{\mu_1}(i) \phi_{\mu_2}(j).
\ea
The only possible nonzero contributions thus are (indices $a$ and $b$ are not summed)
\ba
\langle X_i^a X_j^b \rangle_{D_\mrm{real}} &=&  \frac{\delta^{ab}}{Z[0]} \int Dc \exp \lcb - \frac{i\delta t}{Va^2} \frac{1}{1-p} \int_{\mu\in D_2} \Lambda_{\nu} c^m_{\nu}c^m_{1-\nu} \rcb \times\nn\\
\label{expB3}
& & \times\int_{\nu\in D_\mrm{real}} \lsb \lb c^a_{\mu}\rb^2 + c^a_{\mu}c^a_{1-\mu} + \lb c^a_{1-\mu}\rb^2 \rsb \phi_{\mu}(i) \phi_{\mu}(j).
\ea
Note that on the real line we are not identifying the complex conjugate of $c_\mu^m$ with $c_{1-\mu}^m$, so the expression above cannot be written in terms of absolute values.

Consider now the integral of the middle term; it is of the type
\be
I_\mrm{mid}\coloneqq \int_{c_\mu,c_{1-\mu}} e^{ \Lambda c_\mu c_{1-\mu}}  c_\mu c_{1-\mu} dc_\mu dc_{1-\mu}.
\ee
We perform a variable change to 
\be
s\coloneqq c_\mu, \quad t \coloneqq c_{\mu} c_{1-\mu},
\ee
so that the integral becomes
\be
\label{Imid}
I_\mrm{mid} = \int_{s,t}e^{\Lambda t} s^{-1} t ds dt,
\ee
with the factor of $s^{-1}$ coming from the Jacobian. The integrals over $s$ and $t$ now decouple. With the regularization in Definition \ref{prescription8}, dropping the divergent terms, we have
\be
\int_{-\infty}^0 e^{\Lambda t} t dt= -\frac{1}{\Lambda^2}, \quad  \int_{0}^\infty e^{\Lambda t} t dt= \frac{1}{\Lambda^2},
\ee
so that Eq. \eqref{Imid} vanishes. Similarly, the other two terms in Expression \eqref{expB3} are of the type
\be
\int e^{\Lambda t} t^2 dt,
\ee
so using Definition \ref{prescription8} one more time we have
\be
\int_{-\infty}^0 e^{\Lambda t} t^2 dt= \frac{2}{\Lambda^3}, \quad \int_{0}^\infty e^{\Lambda t} t^2 dt= -\frac{2}{\Lambda^3},
\ee
so that Eq. \eqref{eqisthisvanishing} vanishes.

\end{spacing}

\end{document}